\documentclass[prl,reprint, nofootinbib,amsmath,amssymb, aps, floatfix]{revtex4-1}
\usepackage{dcolumn}% Align table columns on decimal point

\usepackage{bm}% bold math
\usepackage{amssymb,amsmath,amsthm,graphicx,ulem}
\usepackage[linktoc=page]{hyperref}

\usepackage{graphicx,subfigure}
\usepackage{epsfig}
\usepackage{amsmath}
\usepackage{amsfonts}
\usepackage{amssymb}
\usepackage[usenames]{color}
\usepackage{enumerate}

\usepackage[letterpaper,left=2.cm,right=2.cm,top=2.5cm,bottom=2.5cm]{geometry}
\usepackage[T1]{fontenc}

\theoremstyle{remark}

\usepackage{mathtools,xparse}

\DeclarePairedDelimiter{\norm}{\lVert}{\rVert}

\theoremstyle{definition}
\newtheorem{theorem}{Theorem}

\newtheorem{lemma}{Lemma}

\begin{document}
\normalem

\title{Refinement of the Bousso-Engelhardt Area Law}
\author{Fabio Sanches}
\email{fabios@berkeley.edu}
\author{Sean J. Weinberg}
\email{sjweinberg@berkeley.edu}
\affiliation{Department of Physics, University of California, Berkeley, CA 94720, USA}

\bibliographystyle{utcaps}

\begin{abstract}

Past holographic screens are codimension-one surfaces of indefinite signature that are foliated by marginally anti-trapped surfaces called leaves.  Future holographic screens are defined similarly except with marginally trapped leaves.  Bousso and Engelhardt recently showed that the leaves of past and future holographic screens have monotonic area.  We prove a stronger area law that shows that subregions of leaves also have monotonic area.  For every past and future holographic screen, there exists a family of leaf-orthogonal curves called the fibration of the screen.  Any region in a leaf can be translated along the fibration to a leaf of larger area.  Our result states that the area of the subregion grows as it is translated.

\end{abstract}

%\pacs{}
\maketitle

\subsection{Introduction}

Black hole thermodynamics \cite{Hawking:1971tu,Bekenstein:1972tm,Bardeen:1973gs,Bekenstein:1974ax,Bekenstein:1973ur,Hawking:1974rv,Hawking:1974sw} is a critical principle that has guided the development of quantum gravity over the past few decades. In particular, Hawking's area theorem displayed parallels between the area of the event horizon of a black hole and entropy.  This identification of entropy with area is the heart of the holographic behavior \cite{'tHooft:1993gx,Susskind:1994vu} exhibited by gravity. 

Recently, Bousso and Engelhardt \cite{Bousso:2015mqa,Bousso:2015qqa}  proved an area law for surfaces called past and future holographic screens that arise in a more general setting than the spacetimes of black holes. These objects are not defined by the global notion of an event horizon and thus provide an example of ``quasi-locally'' defined surfaces with thermodynamic behavior.  

Holographic screens are well-motivated from considerations in quantum gravity.  The covariant entropy bound suggests   \cite{Bousso:1999xy,Bousso:1999cb} that holographic screens play a role in general spacetimes that is analogous to the AdS boundary\footnote{Ref. \cite{Nomura:2013nya} studies a related construction.
}
 in the context of the AdS/CFT correspondence  \cite{Maldacena:1997re,Witten:1998qj}.  This hypothesis is supported by the recent demonstration that holographic entanglement entropy \cite{Ryu2006,Hubeny2007} can be defined for regions on past and future holographic screens in a way that is consistent with many known properties of entanglement entropy \cite{Sanches:2016sxy}.

Below we show that the Bousso-Engelhardt area law can be refined into a more local form.  The original area law of \cite{Bousso:2015mqa,Bousso:2015qqa} states that preferred codimension-2 surfaces called leaves have monotonic area.  We show that arbitrary subregions of leaves also have monotonic area.  From the point of view of the holographic principle, this provides evidence that degrees of freedom of a holographic description for arbitrary spacetimes are locally distributed and satisfy a local version of the second law of thermodynamics.

 %This also suggests a local second law, and that the organization of degrees of freedom comprising the holographic theory for arbitrary spacetimes is also local.

\subsubsection { Holographic Screens and Area Laws  }

Fix a globally hyperbolic spacetime of dimension $D$ satisfying the genericity conditions stated in \cite{Bousso:2015qqa}. A \textit{past holographic screen}  is a codimension-1 submanifold $H$ of the spacetime that is foliated by marginally anti-trapped surfaces called \textit{leaves}. The foliation into leaves is unique: other splittings of $H$  cannot satisfy the marginally anti-trapped condition.  A \textit{future holographic screen} is instead foliated by marginally trapped surfaces. 

%Here we consider past holographic screens since the generalization to the future case is straightforward. 

The area law of \cite{Bousso:2015mqa,Bousso:2015qqa}  is a statement about the evolution of leaves comprising a past or future holographic screen $H$.  We denote the leaves of $H$ by $\sigma_r$ where $r$ is a smooth parameter. In our notation, we can express the Bousso-Engelhardt area law as the statement that $\norm{\sigma_r}$ is monotonic where $\norm{\cdot}$ denotes the area functional.  By convention, we will always choose the parameter $r$ so that $\norm{\sigma_r}$ is increasing.

On a particular leaf $\sigma$, let $k$ and $l$ denote the two future-directed null vector fields orthogonal $\sigma$.  The condition that $\sigma$ be marginally anti-trapped or marginally trapped can be written in terms of the null expansions $\theta^k$ and $\theta^l$ in the two directions:

\begin{equation}
\label{leafdef}
\begin{aligned}[l]
\begin{split}
\textnormal{\underline{\small{Marginally Anti-Trapped}}} \\
\theta^k = 0 \ \ \ \ \ \ \ \ \ \ \ \ \  \\
 \theta^l > 0 \ \ \ \ \ \ \ \ \ \ \ \ \ 
 \end{split}
\end{aligned}
\quad
\begin{aligned}[r]
\begin{split}
\textnormal{\underline{\small{Marginally Trapped}}} \\
\theta^k = 0 \ \ \ \ \ \ \ \ \ \ \\
 \theta^l < 0\ \ \ \ \ \ \ \ \ \
 \end{split}
\end{aligned}
\end{equation}
In particular, the marginal condition that $\theta^k=0$ means that $\sigma$ is the area-maximizing surface on the geodesic congruence generated by $k$ and $-k$.

We define a vector field $h$ on $H$ by requiring that $h$ is orthogonal to every leaf and by the normalization condition $dr(h) = 1$.  The integral curves of $h$ are called the \textit{fibration}\footnote{Note that $h$ need not have definite signature.  This is the key distinguishing feature between past (and future) holographic screens and related objects including future outer trapping horizons and  dynamical horizons \cite{Hayward:1993wb,Hayward:1997jp,Ashtekar:2003hk,Ashtekar:2005ez}.  Past and future holographic screens can be regarded as a synthesis such ideas with those of \cite{Bousso:1999xy}.} of $H$.  If we extend the definition of $k$ and $l$ to all of $H$, then $h = \alpha l + \beta k$ where $\alpha$ and $\beta$ are smooth functions on $H$.  The Bousso-Engelhardt area law was proven by showing that $\alpha$ never changes sign from which equation \ref{leafdef} implies that leaves have increasing area. 

%If we extend the definition of $k$ and $l$ on the whole screen, we can define a non-vanishing vector field $h$ tangent to the screen and orthogonal to the leaves with a decomposition $h = \alpha l + \beta k$

%Where $\alpha$ and $\beta$ are smooth functions on $H$. The proof of the area law \cite{Bousso:2015mqa,Bousso:2015qqa} relies on showing that $ \alpha >0 $ throughout the screen. Combined with \ref{leafdef}, this implies that leaves have increasing area. Intuitively, we can always ``zig-zag'' to get from one leaf to the next; to first order, moving along $k$ does not change the area, while moving along $l$ increases the area. 

Our area law extends this result as follows. Suppose that $A_0$ is a region in $\sigma_0$.  We can translate $A_0$ to a region $A_r$ in $\sigma_r$ by following the fibration from points in $A_0$ to $\sigma_r$.  We will prove that the area of $A_r$ is increasing.  This conclusion relies on the fact that the area increase associated with zig-zagging along $l$ and $k$ is a first order effect in $r$, while the failure of such a zig-zag procedure to follow the fibration is at most a second-order effect. 

%
%\subsubsection {Relation to the Covariant Holographic Principle}
%While the area law is a statement in general relativity, holographic screens may be important surfaces from the point of view of quantum gravity. In particular, Bousso \cite{Bousso:1999cb} conjectured that theories dual to general spacetimes would live on holographic screens. These are perhaps the most reasonable surfaces from the point of view of the covariant entropy bound \cite{Bousso:1999xy}.
%
%
%  \begin{itemize}
%
%\item Role of holographic screens in raphael's holographic principle
%
%\item second law of thermodynamics in B-E area law
%
%\item Locality of the second Law
%
%
%\end{itemize}

\subsubsection{Relation to the Screen Entanglement Conjecture}  

Holographic entanglement entropy proposals \cite{Ryu2006,Hubeny2007} have recently been conjecturally generalized beyond the context of AdS/CFT by employing past or future holographic screens in arbitrary spacetimes \cite{Sanches:2016sxy}.  The proposed construction is to anchor extremal surfaces to the boundaries of subregions of leaves. The properties of past and future holographic screens are sufficient to ensure that the areas of these extremal surfaces satisfy expected properties of entanglement entropy like strong subadditivity.  The statement that one fourth of the area of such extremal surfaces is in fact the entanglement entropy of a subsystem in a quantum theory holographically defining the spacetime in which the screen lies is called the ``screen entanglement conjecture.''

The area law proven in this paper applies to subregions of leaves, the same objects to which an entanglement entropy-like quantity was assigned in \cite{Sanches:2016sxy}.  Suppose that $A_0$ is a region in the leaf $\sigma_0$ and $A_r$ is the result of translating $A_0$ along the fibration to $\sigma_r$.  Let $S(A_r)$ denote the screen entanglement entropy of $A_r$ as defined above via the extremal surface anchored to $\partial A_r$.  With the exception of cases that are topologically nontrivial, $S(A_r)$ satisfies a ``Page bound'':  $S(A_r) \leq \min(\norm{A_r},\norm{\sigma_r \setminus A_r})$.  Our area law applies to the evolution of the subregions $A_r$ and $A_r^C$ and thus causes the Page bound to become less restrictive whenever $r$ is increased.  This does not prove that $S(A_r)$ increases monotonically.

\subsection{Proof of the Area Law for Subregions}

From here on we will assume that $H$ is a past holographic screen.  Our argument can be modified to the case of a future holographic screen in an obvious way.  Because $H$ is a past screen,
\begin{equation} 
\label{theta_cond}
\begin{split}
\theta^k = 0 \\
 \theta^l > 0.
 \end{split}
\end{equation}
Moreover, we now have $\alpha > 0$ on all of $H$.

To carefully study the evolution of areas of regions in leaves, it is convenient to consider the null surfaces passing through a leaf $\sigma_r$.  First, extend $k$ and $l$ to a tubular neighborhood of $H$  by following along the geodesics generated by $k$ and $l$.  Now let $N_r$ denote the null surface obtained by starting from points on $\sigma_r$ and following the integral curves of $k$ in both the $+k$ and $-k$ directions.  Let $L_r^+$ denote the null surface obtained by starting at $\sigma_r$ and following the integral curves of $l$ only in the $+l$ direction.   

We now fix an (arbitrarily chosen) reference leaf $\sigma_0$.  There exists an $r_0>0$ such that if $0<r<r_0$, it is possible to define a ``zig-zag'' map $f_r:\sigma_0 \to \sigma_r$ as follows.  If $p \in \sigma_0$, follow $L_0^+$ from $p$ along a generator of  $L_0^+$ (i.e. along the integral curve of $l$ that $p$ lies on) until $L_0^+$ intersects a generator of $N_r$.  Then, follow the $N_r$ generator to $\sigma_r$.  Bousso and Engelhardt established that $f_r$ is well-defined for sufficiently small $r$ (this is why we restrict to $r<r_0$).  $f_r$ is, in fact, a diffeomorphism between $\sigma_0$ and $\sigma_r$.  

Considering equation \ref{theta_cond} and the fact that $\alpha>0$, the zig-zag construction of $f_r$ implies that if $A_0$ is a $D-2$ dimensional submanifold of $\sigma_0$,
\begin{equation}
\label{fr_area}
\frac{d}{dr}\Big|_{r=0} \norm{f_r(A_0)} = \int_{A_0} \sqrt{g^{\sigma_0}} \: \alpha \: \theta^l \: >0.
\end{equation}
The area law of Bousso and Engelhardt is obtained in the case where $A_0 = \sigma_0$ because $f_r$ is surjective.

Aside from the case where $A_0 = \sigma_0$, the fact that $\norm{f_r(A_0)}$ is an increasing function of $r$ is an unattractive area law.  One issue is that the definition of the function $f_r$ involves the choice of the reference leaf (i.e. the choice of $r = 0$).  Moreover, the family of regions  $\{ f_r(A_0) \: | \: r\in[0,r_0) \}$ cannot necessarily be extended to all $r$.  

Fortunately, as described above, there is a simpler way to carry subregions from one leaf to the next.  Let $A_{r_1}$ be a $D-2$ dimensional submanifold of the leaf $\sigma_{r_1}$. Define $A_{r_2} \subset \sigma_{r_2}$ by starting from points in $A_{r_1}$ and following along the fibration of $H$ (i.e. the integral curves of $h$) by parameter $r_2 - r_1$.  Note that this procedure gives a well-defined region $A_r \subset \sigma_r$ for the entire range of $r$.  We now prove that $\norm{A_r}$ is an increasing function.
%
%However, we will now state and prove an area law for subregions of leaves that is not afflicted by these criticisms.  Let $A_0$ be any compact $D-2$ dimensional submanifold in the reference leaf $\sigma_0$.  Define $A_r \subset \sigma_r$ by starting from points in $A_0$ and following along the integral curves of $h$ by parameter $r$.  Our new refinement of the Bousso-Engelhardt area law is that $\norm{A_r}$ is a strictly increasing function. 

\begin{figure}
\centering
\includegraphics[width=8cm]{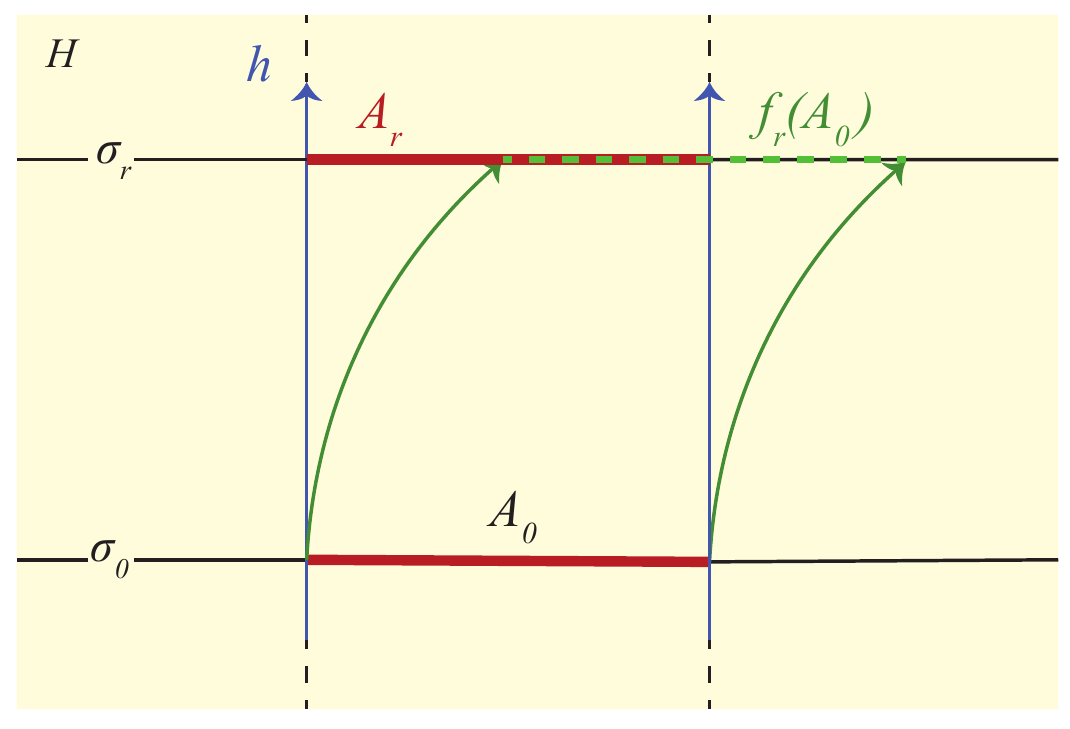}
\caption{We show that $A_r$ has monotonic area by comparing $A_r$ with the region $f_r(A_0)$.  As depicted here, $A_r$ and $f_r(A_0)$ are identical at linear order in $r$.}
\label{comparison}
\end{figure}

First, the following Lemma shows that $f_r$ behaves similarly to $h$-translation for small $r$:
\begin{lemma}
\label{tangent_vec}
If $p_0 \in \sigma_0$, let $\gamma: [0,r_0) \to H$ be the curve on $H$ defined by $\gamma(r) = f_r(p_0)$.  Then, the tangent vector of $\gamma$ at $r=0$ is $h(p_0)$.
\end{lemma}
\begin{proof}
We will begin by introducing a set of convenient coordinates.  Fix a coordinate chart on $\sigma_0$ for a neighborhood of $p_0$.  We denote these coordinates by $x^i$, $i\in \{1,\ldots D-2 \}$ and require that $p_0$ corresponds to the origin of $\bold{R}^{D-2}$.  Extend to coordinates $\{(x^i,r)\}$ on $H$ by following the integral curves of $h$ from $x^i$ by parameter $r$ to reach the point labeled by $(x^i,r)$.  Note that this point will lie in $\sigma_r$.  Finally, extend to coordinates  $\{(x^i,r,z)\}$ by starting from the point $(x^i,r)$ and following the integral curves of $k$ by affine parameter $z$.  Note that $H$ is the $z=0$ hypersurface.

Because $\alpha \neq 0$, we can put $l\big|_H = \frac{1}{\alpha} h - \frac{\beta}{\alpha} k$.  Thus, in the coordinates $(x^i,r,z)$ constructed above, we have
\begin{equation} 
\label{hlk_H}
\begin{split}
&h = (\bold{0},  1,  0) \\
& l \big|_{z=0} = (\bold{0},  \frac{1}{\alpha},  -\frac{\beta}{\alpha}) \\
& k \big|_{z=0} = (\bold{0},  0,  1) 
 \end{split}
\end{equation}
where $\bold{0}$ denotes $D-2$ zeros.  The curve $\gamma(r)$ also takes a simple form in our coordinates: because $f_r$ maps points in $\sigma_0$ to points in $\sigma_r$, we have
\begin{equation} 
\label{gamma_r}
\gamma(r) = (x^i(r), r, 0)
\end{equation}
where $x^i(r)$ is a curve in $\bold{R}^{D-2}$.  Our Lemma will be proven by showing that $\dot{x}^i(0)=0$.

Let $\xi_0(\lambda)$ and $\zeta_r(\lambda)$ denote, respectively, the geodesics generated by $l$ and $k$ from the points $\gamma(0) = (\bold{0},0,0)$ and $\gamma(r) = (x^i(r),r,0)$.  The zig-zag definition of $f_r$ implies that $\xi_0$ and $\zeta_r$ have an intersection: there exist functions $\lambda_1(r)$ and $\lambda_2(r)$ such that
\begin{equation}
\label{intersect}
\xi_0(\lambda_1(r)) = \zeta_r(\lambda_2(r)).
\end{equation}
Meanwhile, equation \ref{hlk_H} implies that
\begin{equation} 
\label{geodesic_series}
\begin{split}
\xi_0(\lambda_1(r)) & = \Big( \bold{0}, \frac{1}{\alpha_0} \lambda_1(r) , -\frac{\beta_0}{\alpha_0} \lambda_1(r) \Big) + O\big(\lambda_1(r)^2 \big)\\
\zeta_r(\lambda_2(r)) & = \Big(x^i(r), r , \lambda_2(r) \Big) + O\big(\lambda_2(r)^2 \big)
 \end{split}
\end{equation}
where $\alpha_0 = \alpha(r=0)$ and $\beta_0 = \beta(r=0)$. Comparing the $r$ and $z$ components of equation \ref{geodesic_series} now gives
\begin{equation} 
\label{lambda_series}
\begin{split}
& \lambda_1(r) = \alpha_0 \: r + O\big(\lambda_1(r)^2, \lambda_2(r)^2 \big) \\
& \lambda_2(r) = -\beta_0 \: r + O\big(\lambda_1(r)^2, \lambda_2(r)^2 \big)
\end{split}
\end{equation}
which then implies that
\begin{equation}
x^i(r) = O\big(\lambda_1(r)^2, \lambda_2(r)^2 \big) = O\big(r^2).
\end{equation}
We conclude that $\dot{x}^i(r=0) = 0$.
\end{proof}

\begin{theorem}
Let $A_0 \subset \sigma_0$ be a $D-2$ dimensional submanifold of $\sigma_0$ and define $A_r$ as the result of translating $A_0$ along the integral curves of $h$ by parameter $r$.  Then, $A_r$ has strictly increasing area.
\end{theorem}
\begin{proof}
Take $r \in [0,r_0)$.  We have
\begin{equation}
\label{sym_diff}
\Big| \norm{f_r(A_0)} - \norm{A_r} \Big| \leq \norm{f_r(A_0)  \: \Delta \: A_r}
\end{equation}
where $\Delta$ denotes the symmetric difference of sets: $A \Delta B = (A \setminus B) \cup (B \setminus A)$.  Now Lemma \ref{tangent_vec} and the compactness of $\sigma_r$ implies that
\[
\frac{d}{dr}\Big|_{r=0} \Big( \norm{f_r(A_0)  \: \Delta \: A_r} \Big)= 0.
\]
Noting that both sides of equation \ref{sym_diff} are nonnegative for all $r$ and are zero at $r=0$, we conclude that
\begin{equation}
\label{deriv_equal}
\frac{d}{dr}\Big|_{r=0} \Bigg( \Big| \norm{f_r(A_0)} - \norm{A_r} \Big| \Bigg) = 0.
\end{equation}
But equation \ref{fr_area} implies that $\norm{f_r(A_0)}$ is increasing at $r=0$ so we must have that $\norm{A_r}$ is also increasing at $r=0$.

While we have only proven that $A_r$ has increasing area at $r=0$, we can define a zig-zag function analogous to $f$ from any reference leaf and repeat all arguments above for any $r$.  Thus, we conclude that $A_r$ has strictly increasing area.  In fact, equations \ref{fr_area} and \ref{deriv_equal} show that
\[
\frac{d}{dr} \norm{A_r} = \int_{A_r} \sqrt{g^{\sigma_r}} \: \alpha \: \theta_l \: >0.
\]

\end{proof}

\noindent {\bf Acknowledgments \ } We are grateful to R. Bousso, M. Moosa, Y. Nomura, and Nico Salzetta for discussions.  The work of FS is supported in part by the DOE NNSA Stewardship Science Graduate Fellowship.  The work of SJW is supported in part by the BCTP Brantley-Tuttle Fellowship for which he would like to extend his gratitude to Lynn Brantley and Douglas Tuttle.


\begin{thebibliography}{99}

%\cite{Hawking:1971tu}
\bibitem{Hawking:1971tu} 
  S.~W.~Hawking,
  ``Gravitational radiation from colliding black holes,''
  Phys.\ Rev.\ Lett.\  {\bf 26}, 1344 (1971).
  doi:10.1103/PhysRevLett.26.1344
  %%CITATION = doi:10.1103/PhysRevLett.26.1344;%%
  %481 citations counted in INSPIRE as of 14 Apr 2016

%\cite{Bekenstein:1972tm}
\bibitem{Bekenstein:1972tm} 
  J.~D.~Bekenstein,
  ``Black holes and the second law,''
  Lett.\ Nuovo Cim.\  {\bf 4}, 737 (1972).
  doi:10.1007/BF02757029
  %%CITATION = doi:10.1007/BF02757029;%%
  %690 citations counted in INSPIRE as of 14 Apr 2016

%\cite{Bekenstein:1973ur}
\bibitem{Bekenstein:1973ur} 
  J.~D.~Bekenstein,
  %``Black holes and entropy,''
  Phys.\ Rev.\ D {\bf 7}, 2333 (1973).
  doi:10.1103/PhysRevD.7.2333
  %%CITATION = doi:10.1103/PhysRevD.7.2333;%%
  %3235 citations counted in INSPIRE as of 17 Apr 2016

%\cite{Bardeen:1973gs}
\bibitem{Bardeen:1973gs} 
  J.~M.~Bardeen, B.~Carter and S.~W.~Hawking,
  ``The Four laws of black hole mechanics,''
  Commun.\ Math.\ Phys.\  {\bf 31}, 161 (1973).
  doi:10.1007/BF01645742
  %%CITATION = doi:10.1007/BF01645742;%%
  %1393 citations counted in INSPIRE as of 13 Apr 2016

%\cite{Bekenstein:1974ax}
\bibitem{Bekenstein:1974ax} 
  J.~D.~Bekenstein,
  ``Generalized second law of thermodynamics in black hole physics,''
  Phys.\ Rev.\ D {\bf 9}, 3292 (1974).
  doi:10.1103/PhysRevD.9.3292
  %%CITATION = doi:10.1103/PhysRevD.9.3292;%%
  %1198 citations counted in INSPIRE as of 14 Apr 2016

%\cite{Hawking:1974rv}
\bibitem{Hawking:1974rv} 
  S.~W.~Hawking,
  ``Black hole explosions,''
  Nature {\bf 248}, 30 (1974).
  doi:10.1038/248030a0
  %%CITATION = doi:10.1038/248030a0;%%
  %2179 citations counted in INSPIRE as of 14 Apr 2016

%\cite{Hawking:1974sw}
\bibitem{Hawking:1974sw} 
  S.~W.~Hawking,
  ``Particle Creation by Black Holes,''
  Commun.\ Math.\ Phys.\  {\bf 43}, 199 (1975)
  Erratum: [Commun.\ Math.\ Phys.\  {\bf 46}, 206 (1976)].
  doi:10.1007/BF02345020
  %%CITATION = doi:10.1007/BF02345020;%%
  %5943 citations counted in INSPIRE as of 14 Apr 2016
  

%\cite{'tHooft:1993gx}
\bibitem{'tHooft:1993gx} 
  G.~'t Hooft,
  ``Dimensional reduction in quantum gravity,''
  Salamfest 1993:0284-296
  [gr-qc/9310026].
  %%CITATION = GR-QC/9310026;%%
  %1712 citations counted in INSPIRE as of 08 Apr 2016
  
  %\cite{Susskind:1994vu}
\bibitem{Susskind:1994vu} 
  L.~Susskind,
  ``The World as a hologram,''
  J.\ Math.\ Phys.\  {\bf 36}, 6377 (1995)
  doi:10.1063/1.531249
  [hep-th/9409089].
  %%CITATION = doi:10.1063/1.531249;%%
  %2067 citations counted in INSPIRE as of 08 Apr 2016

  %\cite{Bousso:2015mqa,Bousso:2015qqa}
  %\cite{Bousso:2015mqa}
\bibitem{Bousso:2015mqa} 
  R.~Bousso and N.~Engelhardt,
  ``New Area Law in General Relativity,''
  Phys.\ Rev.\ Lett.\  {\bf 115}, no. 8, 081301 (2015)
  doi:10.1103/PhysRevLett.115.081301
  [arXiv:1504.07627 [hep-th]].
  %%CITATION = doi:10.1103/PhysRevLett.115.081301;%%
  %7 citations counted in INSPIRE as of 11 Mar 2016

%\cite{Bousso:2015qqa}
\bibitem{Bousso:2015qqa} 
  R.~Bousso and N.~Engelhardt,
  ``Proof of a New Area Law in General Relativity,''
  Phys.\ Rev.\ D {\bf 92}, no. 4, 044031 (2015)
  doi:10.1103/PhysRevD.92.044031
  [arXiv:1504.07660 [gr-qc]].
  %%CITATION = doi:10.1103/PhysRevD.92.044031;%%
  %6 citations counted in INSPIRE as of 11 Mar 2016


%\cite{Maldacena:1997re}
\bibitem{Maldacena:1997re} 
  J.~M.~Maldacena,
  ``The Large N limit of superconformal field theories and supergravity,''
  Int.\ J.\ Theor.\ Phys.\  {\bf 38}, 1113 (1999)
  [Adv.\ Theor.\ Math.\ Phys.\  {\bf 2}, 231 (1998)]
  doi:10.1023/A:1026654312961
  [hep-th/9711200].
  %%CITATION = doi:10.1023/A:1026654312961;%%
  %11552 citations counted in INSPIRE as of 11 Mar 2016

  %\cite{Witten:1998qj}
\bibitem{Witten:1998qj} 
  E.~Witten,
  %``Anti-de Sitter space and holography,''
  Adv.\ Theor.\ Math.\ Phys.\  {\bf 2}, 253 (1998)
  [hep-th/9802150].
  %%CITATION = HEP-TH/9802150;%%
  %7623 citations counted in INSPIRE as of 11 Mar 2016

   
%\cite{Bousso:1999xy}
\bibitem{Bousso:1999xy} 
  R.~Bousso,
  ``A Covariant entropy conjecture,''
  JHEP {\bf 9907}, 004 (1999)
  doi:10.1088/1126-6708/1999/07/004
  [hep-th/9905177].
  %%CITATION = doi:10.1088/1126-6708/1999/07/004;%%
  %522 citations counted in INSPIRE as of 13 Mar 2016

%\cite{Bousso:1999cb}
\bibitem{Bousso:1999cb} 
  R.~Bousso,
  ``Holography in general space-times,''
  JHEP {\bf 9906}, 028 (1999)
  doi:10.1088/1126-6708/1999/06/028
  [hep-th/9906022].
  %%CITATION = doi:10.1088/1126-6708/1999/06/028;%%
  %311 citations counted in INSPIRE as of 12 Mar 2016

%\cite{Nomura:2013nya}
\bibitem{Nomura:2013nya}
  Y.~Nomura, J.~Varela and S.~J.~Weinberg,
  %``Low Energy Description of Quantum Gravity and Complementarity,''
  Phys.\ Lett.\ B {\bf 733} (2014) 126
  doi:10.1016/j.physletb.2014.04.027
  [arXiv:1304.0448 [hep-th]].
  %%CITATION = doi:10.1016/j.physletb.2014.04.027;%%
  %14 citations counted in INSPIRE as of 17 Apr 2016
  
%\cite{Ryu:2006bv}
\bibitem{Ryu2006} 
  S.~Ryu and T.~Takayanagi,
  ``Holographic derivation of entanglement entropy from AdS/CFT,''
  Phys.\ Rev.\ Lett.\  {\bf 96}, 181602 (2006)
  doi:10.1103/PhysRevLett.96.181602
  [hep-th/0603001].
  %%CITATION = doi:10.1103/PhysRevLett.96.181602;%%
  %899 citations counted in INSPIRE as of 11 Mar 2016

%\cite{Hubeny:2007xt}
\bibitem{Hubeny2007} 
  V.~E.~Hubeny, M.~Rangamani and T.~Takayanagi,
  ``A Covariant holographic entanglement entropy proposal,''
  JHEP {\bf 0707}, 062 (2007)
  doi:10.1088/1126-6708/2007/07/062
  [arXiv:0705.0016 [hep-th]].
  %%CITATION = doi:10.1088/1126-6708/2007/07/062;%%
  %342 citations counted in INSPIRE as of 12 Mar 2016

%\cite{Sanches:2016sxy}
\bibitem{Sanches:2016sxy} 
  F.~Sanches and S.~J.~Weinberg,
  ``A Holographic Entanglement Entropy Conjecture for General Spacetimes,''
  arXiv:1603.05250 [hep-th].
  %%CITATION = ARXIV:1603.05250;%%
  
%\cite{Hayward:1993wb}
\bibitem{Hayward:1993wb} 
  S.~A.~Hayward,
  %``General laws of black hole dynamics,''
  Phys.\ Rev.\ D {\bf 49}, 6467 (1994).
  doi:10.1103/PhysRevD.49.6467
  %%CITATION = doi:10.1103/PhysRevD.49.6467;%%
  %372 citations counted in INSPIRE as of 13 Apr 2016
  
  %\cite{Hayward:1997jp}
\bibitem{Hayward:1997jp} 
  S.~A.~Hayward,
  %``Unified first law of black hole dynamics and relativistic thermodynamics,''
  Class.\ Quant.\ Grav.\  {\bf 15}, 3147 (1998)
  doi:10.1088/0264-9381/15/10/017
  [gr-qc/9710089].
  %%CITATION = doi:10.1088/0264-9381/15/10/017;%%
  %285 citations counted in INSPIRE as of 13 Apr 2016

%\cite{Ashtekar:2003hk}
\bibitem{Ashtekar:2003hk} 
  A.~Ashtekar and B.~Krishnan,
  %``Dynamical horizons and their properties,''
  Phys.\ Rev.\ D {\bf 68}, 104030 (2003)
  doi:10.1103/PhysRevD.68.104030
  [gr-qc/0308033].
  %%CITATION = doi:10.1103/PhysRevD.68.104030;%%
  %186 citations counted in INSPIRE as of 14 Mar 2016

%\cite{Ashtekar:2005ez}
\bibitem{Ashtekar:2005ez} 
  A.~Ashtekar and G.~J.~Galloway,
  %``Some uniqueness results for dynamical horizons,''
  Adv.\ Theor.\ Math.\ Phys.\  {\bf 9}, no. 1, 1 (2005)
  doi:10.4310/ATMP.2005.v9.n1.a1
  [gr-qc/0503109].
  %%CITATION = doi:10.4310/ATMP.2005.v9.n1.a1;%%
  %86 citations counted in INSPIRE as of 13 Mar 2016


%\bibitem{Maldacena:1997re} 
%  J.~M.~Maldacena,
%  %``The Large N limit of superconformal field theories and supergravity,''
%  Int.\ J.\ Theor.\ Phys.\  {\bf 38}, 1113 (1999)
%  [Adv.\ Theor.\ Math.\ Phys.\  {\bf 2}, 231 (1998)]
%  doi:10.1023/A:1026654312961
%  [hep-th/9711200].
%  %%CITATION = doi:10.1023/A:1026654312961;%%
%  %11552 citations counted in INSPIRE as of 11 Mar 2016
%  
%  %\cite{Witten:1998qj}
%\bibitem{Witten:1998qj} 
%  E.~Witten,
%  %``Anti-de Sitter space and holography,''
%  Adv.\ Theor.\ Math.\ Phys.\  {\bf 2}, 253 (1998)
%  [hep-th/9802150].
%  %%CITATION = HEP-TH/9802150;%%
%  %7623 citations counted in INSPIRE as of 11 Mar 2016
%
%%\cite{Ryu:2006bv}
%\bibitem{Ryu:2006bv} 
%  S.~Ryu and T.~Takayanagi,
%  %``Holographic derivation of entanglement entropy from AdS/CFT,''
%  Phys.\ Rev.\ Lett.\  {\bf 96}, 181602 (2006)
%  doi:10.1103/PhysRevLett.96.181602
%  [hep-th/0603001].
%  %%CITATION = doi:10.1103/PhysRevLett.96.181602;%%
%  %899 citations counted in INSPIRE as of 11 Mar 2016
%
%%\cite{Lewkowycz:2013nqa}
%\bibitem{Lewkowycz:2013nqa} 
%  A.~Lewkowycz and J.~Maldacena,
%  %``Generalized gravitational entropy,''
%  JHEP {\bf 1308}, 090 (2013)
%  doi:10.1007/JHEP08(2013)090
%  [arXiv:1304.4926 [hep-th]].
%  %%CITATION = doi:10.1007/JHEP08(2013)090;%%
%  %218 citations counted in INSPIRE as of 12 Mar 2016
%
%%\cite{Hubeny:2007xt}
%\bibitem{Hubeny:2007xt} 
%  V.~E.~Hubeny, M.~Rangamani and T.~Takayanagi,
%  %``A Covariant holographic entanglement entropy proposal,''
%  JHEP {\bf 0707}, 062 (2007)
%  doi:10.1088/1126-6708/2007/07/062
%  [arXiv:0705.0016 [hep-th]].
%  %%CITATION = doi:10.1088/1126-6708/2007/07/062;%%
%  %342 citations counted in INSPIRE as of 12 Mar 2016
%  
%  
%%\cite{Bousso:1999xy}
%\bibitem{Bousso:1999xy} 
%  R.~Bousso,
%  %``A Covariant entropy conjecture,''
%  JHEP {\bf 9907}, 004 (1999)
%  doi:10.1088/1126-6708/1999/07/004
%  [hep-th/9905177].
%  %%CITATION = doi:10.1088/1126-6708/1999/07/004;%%
%  %522 citations counted in INSPIRE as of 13 Mar 2016
%
%%\cite{Bousso:1999cb}
%\bibitem{Bousso:1999cb} 
%  R.~Bousso,
%  %``Holography in general space-times,''
%  JHEP {\bf 9906}, 028 (1999)
%  doi:10.1088/1126-6708/1999/06/028
%  [hep-th/9906022].
%  %%CITATION = doi:10.1088/1126-6708/1999/06/028;%%
%  %311 citations counted in INSPIRE as of 12 Mar 2016
%
%  
%  %\cite{Bousso:2015mqa}
%\bibitem{Bousso:2015mqa} 
%  R.~Bousso and N.~Engelhardt,
%  %``New Area Law in General Relativity,''
%  Phys.\ Rev.\ Lett.\  {\bf 115}, no. 8, 081301 (2015)
%  doi:10.1103/PhysRevLett.115.081301
%  [arXiv:1504.07627 [hep-th]].
%  %%CITATION = doi:10.1103/PhysRevLett.115.081301;%%
%  %7 citations counted in INSPIRE as of 11 Mar 2016
%
%%\cite{Bousso:2015qqa}
%\bibitem{Bousso:2015qqa} 
%  R.~Bousso and N.~Engelhardt,
%  %``Proof of a New Area Law in General Relativity,''
%  Phys.\ Rev.\ D {\bf 92}, no. 4, 044031 (2015)
%  doi:10.1103/PhysRevD.92.044031
%  [arXiv:1504.07660 [gr-qc]].
%  %%CITATION = doi:10.1103/PhysRevD.92.044031;%%
%  %6 citations counted in INSPIRE as of 11 Mar 2016
%
%%\cite{Ashtekar:2003hk}
%\bibitem{Ashtekar:2003hk} 
%  A.~Ashtekar and B.~Krishnan,
%  %``Dynamical horizons and their properties,''
%  Phys.\ Rev.\ D {\bf 68}, 104030 (2003)
%  doi:10.1103/PhysRevD.68.104030
%  [gr-qc/0308033].
%  %%CITATION = doi:10.1103/PhysRevD.68.104030;%%
%  %186 citations counted in INSPIRE as of 14 Mar 2016
%
%%\cite{Ashtekar:2005ez}
%\bibitem{Ashtekar:2005ez} 
%  A.~Ashtekar and G.~J.~Galloway,
%  %``Some uniqueness results for dynamical horizons,''
%  Adv.\ Theor.\ Math.\ Phys.\  {\bf 9}, no. 1, 1 (2005)
%  doi:10.4310/ATMP.2005.v9.n1.a1
%  [gr-qc/0503109].
%  %%CITATION = doi:10.4310/ATMP.2005.v9.n1.a1;%%
%  %86 citations counted in INSPIRE as of 13 Mar 2016
%
% %\cite{Headrick:2007km}
%\bibitem{Headrick:2007km} 
%  M.~Headrick and T.~Takayanagi,
%  %``A Holographic proof of the strong subadditivity of entanglement entropy,''
%  Phys.\ Rev.\ D {\bf 76}, 106013 (2007)
%  doi:10.1103/PhysRevD.76.106013
%  [arXiv:0704.3719 [hep-th]].
%  %%CITATION = doi:10.1103/PhysRevD.76.106013;%%
%  %127 citations counted in INSPIRE as of 14 Mar 2016
%
%
%%\cite{Wall:2012uf}
%\bibitem{Wall:2012uf} 
%  A.~C.~Wall,
%  %``Maximin Surfaces, and the Strong Subadditivity of the Covariant Holographic Entanglement Entropy,''
%  Class.\ Quant.\ Grav.\  {\bf 31}, no. 22, 225007 (2014)
%  doi:10.1088/0264-9381/31/22/225007
%  [arXiv:1211.3494 [hep-th]].
%  %%CITATION = doi:10.1088/0264-9381/31/22/225007;%%
%  %72 citations counted in INSPIRE as of 11 Mar 2016
%
%%\cite{Engelhardt:2013tra}
%\bibitem{Engelhardt:2013tra} 
%  N.~Engelhardt and A.~C.~Wall,
%  %``Extremal Surface Barriers,''
%  JHEP {\bf 1403}, 068 (2014)
%  doi:10.1007/JHEP03(2014)068
%  [arXiv:1312.3699 [hep-th]].
%  %%CITATION = doi:10.1007/JHEP03(2014)068;%%
%  %24 citations counted in INSPIRE as of 11 Mar 2016
%
%%\cite{Hartman:2013qma}
%\bibitem{Hartman:2013qma} 
%  T.~Hartman and J.~Maldacena,
%  %``Time Evolution of Entanglement Entropy from Black Hole Interiors,''
%  JHEP {\bf 1305}, 014 (2013)
%  doi:10.1007/JHEP05(2013)014
%  [arXiv:1303.1080 [hep-th]].
%  %%CITATION = doi:10.1007/JHEP05(2013)014;%%
%  %114 citations counted in INSPIRE as of 11 Mar 2016
%
%%\cite{Page:1993df}
%\bibitem{Page:1993df} 
%  D.~N.~Page,
%  %``Average entropy of a subsystem,''
%  Phys.\ Rev.\ Lett.\  {\bf 71}, 1291 (1993)
%  doi:10.1103/PhysRevLett.71.1291
%  [gr-qc/9305007].
%  %%CITATION = doi:10.1103/PhysRevLett.71.1291;%%
%  %225 citations counted in INSPIRE as of 15 Mar 2016
%
%%\cite{Nomura:2011dt}
%\bibitem{Nomura:2011dt} 
%  Y.~Nomura,
%  %``Physical Theories, Eternal Inflation, and Quantum Universe,''
%  JHEP {\bf 1111}, 063 (2011)
%  doi:10.1007/JHEP11(2011)063
%  [arXiv:1104.2324 [hep-th]].
%  %%CITATION = doi:10.1007/JHEP11(2011)063;%%
%  %45 citations counted in INSPIRE as of 16 Mar 2016
%
%  %\cite{Nomura:2011rb}
%\bibitem{Nomura:2011rb} 
%  Y.~Nomura,
%  %``Quantum Mechanics, Spacetime Locality, and Gravity,''
%  Found.\ Phys.\  {\bf 43}, 978 (2013)
%  doi:10.1007/s10701-013-9729-1
%  [arXiv:1110.4630 [hep-th]].
%  %%CITATION = doi:10.1007/s10701-013-9729-1;%%
%  %32 citations counted in INSPIRE as of 15 Mar 2016
%  
%  %\cite{Nomura:2013nya}
%\bibitem{Nomura:2013nya} 
%  Y.~Nomura, J.~Varela and S.~J.~Weinberg,
%  %``Low Energy Description of Quantum Gravity and Complementarity,''
%  Phys.\ Lett.\ B {\bf 733}, 126 (2014)
%  doi:10.1016/j.physletb.2014.04.027
%  [arXiv:1304.0448 [hep-th]].
%  %%CITATION = doi:10.1016/j.physletb.2014.04.027;%%
%  %13 citations counted in INSPIRE as of 13 Mar 2016
%  
%%\cite{Nomura:2014voa}
%\bibitem{Nomura:2014voa} 
%  Y.~Nomura, F.~Sanches and S.~J.~Weinberg,
%  %``Relativeness in Quantum Gravity: Limitations and Frame Dependence of Semiclassical Descriptions,''
%  JHEP {\bf 1504}, 158 (2015)
%  doi:10.1007/JHEP04(2015)158
%  [arXiv:1412.7538 [hep-th]].
%  %%CITATION = doi:10.1007/JHEP04(2015)158;%%
%  %6 citations counted in INSPIRE as of 13 Mar 2016
%
%%\cite{Nomura:2014woa}
%\bibitem{Nomura:2014woa} 
%  Y.~Nomura, F.~Sanches and S.~J.~Weinberg,
%  %``Black Hole Interior in Quantum Gravity,''
%  Phys.\ Rev.\ Lett.\  {\bf 114}, 201301 (2015)
%  doi:10.1103/PhysRevLett.114.201301
%  [arXiv:1412.7539 [hep-th]].
%  %%CITATION = doi:10.1103/PhysRevLett.114.201301;%%
%  %8 citations counted in INSPIRE as of 13 Mar 2016
%
%%\cite{Almheiri:2012rt}
%\bibitem{Almheiri:2012rt} 
%  A.~Almheiri, D.~Marolf, J.~Polchinski and J.~Sully,
%  %``Black Holes: Complementarity or Firewalls?,''
%  JHEP {\bf 1302}, 062 (2013)
%  doi:10.1007/JHEP02(2013)062
%  [arXiv:1207.3123 [hep-th]].
%  %%CITATION = doi:10.1007/JHEP02(2013)062;%%
%  %438 citations counted in INSPIRE as of 16 Mar 2016
%  
%  %\cite{Faulkner:2013ana}
%\bibitem{Faulkner:2013ana} 
%  T.~Faulkner, A.~Lewkowycz and J.~Maldacena,
%  %``Quantum corrections to holographic entanglement entropy,''
%  JHEP {\bf 1311}, 074 (2013)
%  doi:10.1007/JHEP11(2013)074
%  [arXiv:1307.2892 [hep-th]].
%  %%CITATION = doi:10.1007/JHEP11(2013)074;%%
%  %101 citations counted in INSPIRE as of 16 Mar 2016
%  
%  %\cite{Bousso:2015eda}
%\bibitem{Bousso:2015eda} 
%  R.~Bousso and N.~Engelhardt,
%  %``Generalized Second Law for Cosmology,''
%  Phys.\ Rev.\ D {\bf 93}, no. 2, 024025 (2016)
%  doi:10.1103/PhysRevD.93.024025
%  [arXiv:1510.02099 [hep-th]].
%  %%CITATION = doi:10.1103/PhysRevD.93.024025;%%
%  
  \end{thebibliography}
\end{document}